\documentclass[a4paper,11pt]{article}
\usepackage{amsmath,amssymb,amsthm}
\usepackage{amsfonts}
\usepackage{eucal}
\numberwithin{equation}{section}

\renewcommand{\a}{\alpha}
\renewcommand{\b}{\beta}

\renewcommand{\d}{\delta}

\newcommand{\f}{\varphi}

\newcommand{\s}{\sigma}
\newcommand{\x}{\xi}

\renewcommand{\L}{\Lambda}
\newcommand{\co}{\mathbb{C}}
\newcommand{\re}{\mathbb{R}}
\newcommand{\ze}{\mathbb{Z}}

\def\pa{\partial}

\newcommand{\supp}{\mathrm{{supp}}}

\newcommand{\fourier}{\mathcal{F}} 
\newcommand{\torus}{\mathbb{T}}

\newcommand{\bigpare}[1]{\bigl(#1\bigr)}
\newcommand{\biggpare}[1]{\biggl(#1\biggr)}
\newcommand{\Bigpare}[1]{\Bigl(#1\Bigr)}

\newcommand{\bigset}[2]{\bigl\{#1\bigm|#2\bigr\}}

\newcommand{\norm}[1]{\| #1 \|}
\newcommand{\bignorm}[1]{\bigl\| #1 \bigr\|}
\newcommand{\Bignorm}[1]{\Bigl\| #1 \Bigr\|}

\newcommand{\abs}[1]{| #1 |}
\newcommand{\bigabs}[1]{\bigl| #1 \bigr|}

\newtheorem{thm}{Theorem}[section]

\theoremstyle{definition}

\newtheorem{ass}{Assumption}

\theoremstyle{remark}

\title{Remarks on discrete Dirac operators and their continuum limits\footnote{The author thanks  
Tohru Koma for informing him about the paper by Susskind, and Shinya Aoki for valuable discussion and comments. He also thanks Yukimi Goto and  Yukihide Tadano for stimulating discussions. 
The research was partly supported by JSPS Kakenhi Grant Number 21K03276.} }
\author{Shu Nakamura\footnote{Department of Mathematics, Faculty of Sciences, Gakushuin University, 1-5-1, Mejiro, Toshima, Tokyo, Japan 171-8588, \texttt{shu.nakamura@gakushuin.ac.jp}}}

\begin{document}
\maketitle

\begin{abstract}
We discuss possible definitions of discrete Dirac operators, and discuss their continuum limits. 
It is well-known in the lattice field theory that the straightforward discretization of the Dirac 
operator introduces unwanted spectral subspaces, and it is known as {\em the fermion doubling}. 
In oder to overcome this difficulty, two methods were proposed. The first one is to introduce 
a new term, called {\em the Wilson term}, and the second one is {\em the KS-fermion model} or
{\em the staggered fermion model}. 
We discuss mathematical formulations of these, and study their continuum limits. 
\end{abstract}

%%%%%%%%%%%%%%%%%%%%%%%% SECTION 1%%%%%%%%%%%%%%%%%%%%%%%%%%
\section{Introduction}\label{sec-intro}

In a recent paper by Cornean, Garde and Jensen \cite{C-G-J-2}, they studied continuum limit of discretized 
Dirac operators in the sense of norm resolvent convergence, 
and they found that they do not converges to the (usual) Dirac operators. 
They found that if one add another term, then these operators converges to the Dirac 
operators. This corresponds to \textit{the Wilson term} in the lattice field theory. 
We discuss this method briefly, and then discuss another method, \textit{the KS-fermion} 
(or \textit{the staggered fermion}) model, 
which is mathematically ingenious and interesting in itself. Thus this note is partly a survey 
of these methods, but they are rigorously reformulated in relatively general settings, and we prove 
some new results on their continuum limits. 

The continuum limit of quantum Hamiltonian on the square lattice in the sense 
of (generalized) norm resolvent convergence was studied by Nakamura and Tadano \cite{NaTad}, 
and several research has been appeared based on the idea of the norm resolvent convergence 
(see also \cite{Post} for the concept of generalized resolvent convergence). 
Cornean, Garde and Jensen \cite{C-G-J} studied the convergence for more general Fourier multipliers, 
and Exner, Nakamura and Tadano \cite{ExNaTad} considered continuum limit for quantum graph 
Hamiltonians. As mentioned above, Cornean, Garde and Jensen \cite{C-G-J-2} considered 
the continuum limit of discretized Dirac operators, and found difficulty to show the convergence 
to the continuous Dirac operators. See also Schmidt and Umeda \cite{Schmidt-Umeda} and 
Isozaki and Jensen \cite{Isozaki-Jensen} for closely related results. 

It turned out that such difficulty was widely known in the lattice field (gauge) theory (see, 
e.g., \cite{Aoki}, \cite{Rothe}), and it is generally called \emph{the fermion doubling}. 
There are two standard methods to avoid the problem. The first one is adding an additional 
term to the Hamiltonian (or the Lagrangian), and it is called the Wilson term. 
The other method is called the KS-fermion model after Kogut and Susskind (\cite{Susskind}
and \cite{KS}). We try to reformulate these methods, especially the KS-fermion method 
so that it is appropriate to study the continuum limit in the norm resolvent sense, and prove the 
convergence of the continuum limit. 

The paper is constructed as follows. In Section~\ref{sec-pre}, we prepare several basic tools. 
At first we explain the notations concerning the square lattices, function spaces, Fourier transforms, 
and several kinds of difference operators. Then in Subsection~\ref{subsec-embedding}, 
we introduce a embedding operator from function space on the lattice to the Lebesgues space 
on the Euclidean space, which is necessary to study the continuum limit following \cite{NaTad}. 
In Subsection~\ref{subsec-freeDirac}, we recall the definition of the Dirac operators on the 
Euclidean spaces. In Section~\ref{sec-NGDiscretization}, we consider the discretization of
the Dirac operator using the symmetric difference operators, and explain why it is not 
appropriate to consider the continuum limit. In Section~\ref{sec-Wilson}, we introduce 
the Wilson term, and show the convergence of the continuum limit in the norm resolvent sense
for Hamiltonians with the Wilson term under suitable conditions. Section~\ref{sec-KSmodel} is devoted 
to the discussion of the KS-fermion model. We introduce the one-component KS-Hamiltonian 
on $d$-dimensional lattice (with fermion doubling problem), 
and then transform it to a $2^d$ components operator without the fermion doubling problem in Subsection~\ref{subsec-KS1}. We briefly examine the spectral properties of this operator in 
Subsection~\ref{subsec-KS2}, and prove the convergence to a continuum limit in Subsection~\ref{subsec-KS3}. Here the number of components, $2^d$, can be higher than 
those of the standard Dirac operator on $\re^d$. 
We discuss the model for the dimensions 1, 2 and 3 in Section~\ref{sec-examples}, 
and show that for $d=1$ the model is appropriate (and in fact studied 
in \cite{C-G-J-2} already), and for $d=2$ and $3$, the continuum limit is decomposed 
to a direct sum of two standard Dirac operators. 

%%%%%%%%%%%%%%%%
%%%%%%%%%%%%%%%%
\section{Preliminaries}\label{sec-pre}

%%%%%%%%%%%%%%%%
\subsection{Notations}\label{subsec-notations}
We denote the square lattice in $\re^d$ with the lattice spacing $h>0$ by 
$h\ze^d=\bigset{hn}{n\in\ze^d}$. 
Let $\{e_j\}_{j=1}^d$ be the standard basis of $\re^d$, i.e., 
$e_j=(\d_{j,k})_{k=1}^d\in\ze^d$, $j=1,\dots, d$, where $\d_{j,k}$ denotes the Kronecker 
delta symbol. The basis (or the generators) of $h\ze^d$ is given by $\{he_1,\dots, he_d\}$. 
We recall the dual space (or the dual module) of $h\ze^d$ is given by 
$h^{-1}\torus^d = \re^d/(h^{-1}\ze^d)$. We note the dual lattice of $h\ze^d$ is $h^{-1}\ze^d$, and 
hence the inner product $z\cdot\x$ is well-defined modulo $\ze$ for $z\in h\ze^d$, $\x\in h^{-1}\torus^d$. 

We denote the standard $L^2$ space on the $d$ dimensional Euclidean space by $L^2(\re^d)$. 
We use the square summable function space on $h\ze^d$, 
namely $\ell^2(h\ze^d)$, and we use the norm defined by 
\[
\norm{u}_{\ell^2(h\ze^d)}^2 = h^d \sum_{z\in h\ze^d} |u(z)|^2, \quad u\in\ell^2(h\ze^d).
\]

We denote the Fourier transform on $\re^d$ by 
\[
\fourier u(\x) =\int_{\re^d} e^{-2\pi i x\cdot\x}u(x)dx \quad \text{for }u\in L^1(\re^d), \x\in\re^d, 
\]
and the inverse Fourier transform by $\fourier^*$. 
On the lattice $h\ze^d$, the Fourier transform $F_h$ : $\ell^2(h\ze^d)\to L^2(h^{-1}\torus^d)$ 
is defined by
\[
F_hu (\x) =h^d \sum_{z\in h \ze^d} e^{-2\pi i z\cdot\x} u(z), \quad  \x\in h^{-1}\torus^d. 
\]
for $u\in\ell^2(h\ze^d)$. 
$F_h$ is unitary, and the inverse is given by 
\[
F_h^* v(z) =\int_{h^{-1}\torus^d} e^{2\pi i z\cdot\x} v(\x)d\x, \quad z\in h\ze^d
\]
for $v\in L^2(h^{-1}\torus^d)$. 

The partial differential operator on $\re^d$, or the momentum operator is denoted by
\[
D_j =\frac{1}{i} \frac{\pa}{\pa x_j}, \quad j=1,\dots, d.
\]
On the lattice $h\ze^d$, we set the symmetric difference operators 
\[
D^S_{h;j}u(z) = \frac{1}{2ih}(u(z+h e_j)-u(z-h e_j)), \quad j=1,\dots, d, z\in h \ze^d, 
\]
as an approximation of $D_j$ on $h\ze^d$, where $u\in \ell^2(h\ze^d)$. 
We also write the forward and backward difference operators by 
\[
D^\pm_{h;j}u(z) = \pm \frac{1}{ih}(u(z\pm he_j)-u(z)), \quad z\in h\ze^d
\]
for $u\in \ell^2(h\ze^d)$. 

%%%%%%%%%
\subsection{Embedding of $\ell^2(h\ze^d)$ into $L^2(\re^d)$}\label{subsec-embedding}

We need an embedding operator $J_h: \ell^2(h\ze^d)\to L^2(\re^d)$ 
when we consider the continuum limit. We employ the following operators 
(\cite{NaTad},  see also \cite{C-G-J}). 

We need a function $\f\in\mathcal{S}(\re^d)$ such that $\{\f(\cdot-n)\mid n\in\ze^d\}$ 
is an orthonormal system in $L^2(\re^d)$. It is well-known that this condition is equivalent to 
\begin{equation}\label{eq-orthonormal}
\sum_{n\in\ze^d} |\hat\f(\x+n)|^2 =1\quad \text{for }\x\in\re^d, 
\end{equation}
where $\hat\f=\mathcal{F}\f$. We then set, for $z\in h\ze^d$, 
\[
\f_{h;z}(x)= \f(h^{-1}(x-z)), \quad x\in\re^d, 
\]
and we define 
\[
J_h u(x) =\sum_{z\in h\ze^d} u(z)\f_{h;z}(x), \quad x\in\re^d.
\]
It is easy to see $J_h$ is an isometry from $\ell^2(h\ze^d)$ to $L^2(\re^d)$ provided 
$\f$ satisfies \eqref{eq-orthonormal}, and the adjoint operator is given by 
\[
J_h^* v(z)= h^{-d} \int_{\re^d} \overline{\f_{h;z}(x)}v(x) dx, \quad z\in h\ze^d,
\]
where $v\in L^2(\re^d)$. (We remark that our notations are slightly different from \cite{NaTad}. 
In particular, $J_h=P_h^*$ in \cite{NaTad}). 
In the following, we always suppose 

\begin{ass}\label{ass-orthonormal-basis}
$\f$ satisfies the condition \eqref{eq-orthonormal}, and $\supp[\hat \f]\subset (-1,1)^d$. 
\end{ass}

We note there exists various such $\f$'s, and we simply choose one and fix it here. 
See \cite{NaTad} for the detail. 

%%%%%%%%%
\subsection{Free Dirac operators}\label{subsec-freeDirac}
We recall the definition of the Dirac operators on $\re^d$. See, e.g., Thaller \cite{Thaller} for the 
survey on Dirac operators. For simplicity, we mainly 
discuss the free operators without perturbations here. 
Let $\a_1,\dots,\a_d$ and $\b$ be a set of $N\times N$ Hermitian matrices such that 
\[
\a_i\a_j+\a_j\a_i=0, \quad \a_i\b +\b\a_i =0, \quad i\neq j, 
\]
and $\a_j^2=\b^2=\mathbf{1}_N$, where $N\in 2\mathbb{N}$ and $\mathbf{1}_N$ denotes the 
$N\times N$ identity matrix. Typical choices for $d=1,2,3$ are as follows: We denote a set of 
Pauli matrices by
\[
\s_1=\begin{pmatrix} 0 & 1 \\ 1 & 0 \end{pmatrix}, \quad 
\s_2=\begin{pmatrix} 0 & -i \\ i & 0 \end{pmatrix}, \quad
\s_3=\begin{pmatrix} 1 & 0 \\ 0 & -1 \end{pmatrix}. 
\]
For $d=1$, we set $N=2$ and $\a_1=\s_1$ and $\b=\s_3$. 
For $d=2$, we set $N=2$ and $\a_1=\s_1$, $\a_2=\s_2$ and $\b=\s_3$. 
For $d=3$, we set $N=4$ and 
\[
\a_j=\begin{pmatrix} 0 & \s_j \\ \s_j & 0 \end{pmatrix}  \text{ for } j=1,2,3; \quad 
\b=\begin{pmatrix} \mathbf{1}_2 & 0 \\ 0 & -\mathbf{1}_2 \end{pmatrix}.
\]
We then define the (free) Dirac operator by 
\[
H_0= \sum_{j=1}^d D_j \a_j + m \b \quad \text{on } [L^2(\re^d)]^{\oplus N}, 
\]
where $D_j=-i\pa/\pa{x_j}$, $j=1,\dots, d$, and $m\geq 0$. 

It is easy to see by the anti-commuting properties, 
\[
H_0^2 =\sum_{j=1}^d D_j^2\a_j^2 +m^2\b^2 = (-\triangle +m^2)\mathbf{1}_N,
\]
and hence $H_0$ is elliptic. It is also straightforward to show $H_0$ is self-adjoint with 
$\mathcal{D}(H_0)=[H^1(\re^d)]^{\oplus N}$. 
We note  
\[
\fourier H_0 \fourier^* u(\x) = \hat H_0(\x)u(\x), \quad \text{where } 
\hat H_0(\x)=\sum_{j=1}^d 2\pi\x_j\a_j +m\b.
\]
Since $\hat H_0(\x)^2=|2\pi\x|^2+m^2$, it can be shown that the eigenvalues of $\hat H_0(\x)$ are 
$\pm \sqrt{|2\pi\x|^2+m^2}$ with multiplicities $N/2$ each.

%%%%%%%%%%%
%%%%%%%%%%%
\section{Straightforward discretization of Dirac operators and the fermion doubling}
\label{sec-NGDiscretization}

We now discretize the Dirac operators on $h\ze^d$.
Using $D^S_{h;j}$, we may define the discretized Dirac operator by
\[
H^S_{0;h} = \sum_{j=1}^d D^S_{h,j}\a_j +m\b\quad\text{on }[\ell^2(h\ze^d)]^{\oplus N},
\]
which is a bounded symmetric operator. 
The symbol of $H_{0;h}^S$, $\hat H_{0;h}^S(\x)$,  is defined by 
\[
\hat H_{0;h}^S(\x)v(\x)= F_h H_{0;h}^S F_h^* v(\x) = \biggpare{\sum_{j=1}^d h^{-1}\sin(2\pi h\x_j)\a_j +m\b}v(\x)
\]
for $v\in [L^2(h^{-1}\torus^d)]^{\oplus N}$. The eigenvalues of $\hat H_{0;h}^S(\x)$ are given by 
\[
E_{\pm,h}(\x) =\pm \biggpare{ h^{-2} \sum_{j=1}^d \sin^2(2\pi h\x_j)+m^2}^{1/2}, 
\quad \x\in h^{-1}\torus^d. 
\]

We note the eigenvalues of $H_0$ are given by $E_{\pm}(\x)=\pm \sqrt{|2\pi \x|^2+m^2}$ 
and $|E_\pm(\x)|$  have only one critical point (local minimal point) $\x=0$ with the minimal value
$m$ if $m>0$. If $m=0$, $E_\pm(\x)=0$  only at $\x=0$. 
On the other hand, $|E_{\pm,h}(\x)|$ has $2^d$ local minimal points: $\{0,(2h)^{-1}\}^d$ 
(with the minimal value $m$) if $m>0$, and the $2^d$ zero points: $\{0,(2h)^{-1}\}^d$ if $m=0$. 
Hence when $h\to 0$, the resolvent of $H_{0;h}^S$ converges to the direct sum of $2^d$ copies of the 
resolvent of $H_0$ with suitable identification. 
In particular, $H_{0;h}^S$ cannot converges to the resolvent of $H_0$ in the norm resolvent sense
(see \cite{C-G-J-2} Theorem~4.7, Theorem~5.7). 
In physics terminology, this implies $H_{0;h}^S$ describes $2^d$ different fermion particles, 
and thus this phenomenon is called \textit{the fermion doubling} (\cite{Aoki, Rothe}). 
For this reason, $H_{0;h}^S$ is not considered a reasonable discretization of the Dirac operator. 

%%%%%%%%%%%
%%%%%%%%%%%
\section{The Wilson term}\label{sec-Wilson}
One standard procedure to avoid the fermion doubling is adding a term of the form 
\[
S_W=\rho (-\triangle_h)\b,
\]
to the Hamiltonian, where $\triangle_h$ is the standard difference Laplacian defined by 
\[
-\triangle_h u(z)=h^{-2}\sum_{j=1}^d (2u(z)-u(z+he_j)-u(z-he_j)), 
\quad z\in h\ze^d, 
\]
and $\rho>0$ is a small coupling constant. $S_W$ is called the Wilson term (see \cite{Aoki},
\cite{Rothe} Section~4.3). 
We set
\[
\tilde H_{0;h} = H_{0;h}^S+S_W.
\]
If $\rho\to 0$ and $\rho h^{-2} \to\infty$ as $h\to 0$, one can show that $\tilde H_{0;h}$ 
converges to $H_0$ in the norm resolvent sense as $h\to 0$ (Theorem~\ref{thm-Wilson-convergence}. 
See also \cite{C-G-J-2} Section~4.1 and Section~5.1, where the coupling constant is chosen 
as $\rho=h$). 

The Wilson term destroys the fermion doubling for the following simple reason. The symbol of 
$\tilde H_{0;h}$ is given by 
\[
\hat H_{0;h}(\x)=\sum_{j=1}^d h^{-1}\sin(2\pi h\x_j)\a_j +\biggpare{m+\rho\sum_{j=1}^d 2 h^{-2}(1-\cos(2\pi h\x_j))}\b, 
\]
and its eigenvalues are given by 
\[
\tilde E_{0;\pm}(\x)= \pm \biggpare{ \sum_{j=1}^d h^{-2}\sin^2(2\pi h\x_j) 
+\Bigpare{m+\rho\sum_{j=1}^d 2 h^{-2}(1-\cos(2\pi h\x_j))}^2}^{1/2}.
\]
These eigenvalues $|\tilde E_{0;\pm}|$ still have $2^d$ local minimal points in the case $m>0$, but 
 $|\tilde E_{0;\pm}|\geq m+\rho h^{-2}$ at these local minima, 
except for $\x=0$, and they diverges to $+\infty$ as $h\to 0$. 
In the case $m=0$, these eigenvalues are at least of order $O(\rho h^{-2})$ away from 
any neighborhood of $\x=0$, and hence the absolute values of eigenvalues diverges to $+\infty$
as $h\to 0$. On the other hand, if $\rho\to 0$, the Wilson term is negligible in a neighborhood of 
$\x=0$ as $h\to 0$. Specifically, we have 

\begin{thm}\label{thm-Wilson-convergence}
Suppose $\rho\to 0$ and $\rho h^{-2}\to\infty$ as $h\to 0$. Then for $z\in\co\setminus\re$, 
\[
\Bignorm{(H_0-z)^{-1} - J_h (\tilde H_{0;h}-z)^{-1}J_h^*}_{\mathcal{B}(L^2(\re^d))}\to 0,
\quad \text{as }h\to 0.
\]
Since $J_h$ is an isometry, this also implies 
\[
\Bignorm{J_h^*(H_0-z)^{-1}J_h - (\tilde H_{0;h}-z)^{-1}}_{\mathcal{B}(\ell^2(h\ze^d))}\to 0,
\quad \text{as }h\to 0.
\]
\end{thm}

\begin{proof}
The proof is essentially the same as the argument in \cite{NaTad} Section~2, and 
\cite{C-G-J-2} Sections~4.1 and 5.2. We only sketch the argument. We follow the 
notations of \cite{NaTad}, and we write $Q_h = F_h J_h^* \fourier^*$
and $\hat H_0=\fourier H_0 \fourier^* = \hat H_0(\x)\cdot$. Then we have
\begin{equation}\label{eq-Wilson-1}
\bignorm{(1-J_h J_h^*)(H_0-z)^{-1}}
= \bignorm{(1-Q_h^* Q_h)(\hat H_0-z)^{-1}}\leq Ch
\end{equation}
by the same proof as in Lemma~2.2, \cite{NaTad}. Now we note 
\[
\bigabs{\hat H_{0;h}(\x)-\hat H_0(\x)}\leq C h^2|\x|^3 + C |\rho|\,|\x|^2
\]
for $\x\in\re^d$, where $\abs{\cdot}$ in the left hand side denotes the operator norm 
in $\co^N$. We also note
\[
\bigabs{\bigpare{\hat H_0(\x)-z}^{-1}}\leq |\x|^{-1}
\]
and
\[
\bigabs{\bigpare{\hat H_{0;h}(\x)-z}^{-1}}\leq \max_\pm \bigabs{\tilde E_{0;\pm}(\x)-z}^{-1}
\leq C(|\x|^{-1} + |\rho|^{-1}h^2)
\]
on the support of $\hat\f(h\x)$. Combining these, we have 
\begin{align*}
\bigabs{(\hat H_0(\x)-z)^{-1}-(\hat H_{0;h}(\x)-z)^{-1}}
&\leq Ch^2|\x| + C|\rho|+C|\rho|^{-1}h^4|\x|^2\\
&\leq Ch + C|\rho|+C|\rho|^{-1}h^2
\end{align*}
on the support of $\hat\f(h\x)$. This implies 
\[
\bignorm{(\tilde H_{0;h}-z)^{-1}J_h^* -J_h^* (H_0-z)^{-1}} \leq C(h+|\rho|+|\rho h^{-2}|^{-1})
\]
as well as Lemma~2.3 of \cite{NaTad}. Combining this with \eqref{eq-Wilson-1}, we arrive at the conclusion. 
\end{proof}

%%%%%%%%%%%
%%%%%%%%%%%
\section{The KS-fermion model} \label{sec-KSmodel}

%%%%%%%%%%
\subsection{The construction of the KS-Hamiltonian}\label{subsec-KS1}

Here we describe an interpretation of an idea by Susskind \cite{Susskind} (see also 
Kogut-Susskind \cite{KS} and \cite{Rothe} Section~4.4), which is called the KS-fermion (or the staggered fermion) model.  We write
\[
s_j(n)=\sum_{k=1}^j n_k, \quad \text{for }n\in\ze^d, j=1,\dots,d, 
\]
and we also set $s_0(n)=0$. We define operators $X_{h;j}$ and $Y_h$ on $\ell^2(h\ze^d)$ by 
\begin{align*}
X_{h;j} u(z)&= (-1)^{s_{j-1}(z/h)}D^S_{h;j}u(z), \quad z \in h\ze^d,\\
Y_hu(z) &= (-1)^{s_d(z/h)}u(z), \quad z \in h\ze^d,
\end{align*}
where $u\in \ell^2(h\ze^d)$ and $j=1,\dots, d$. By direct computations, we can easily show
\begin{align*}
X_{h;j}X_{h;k}+X_{h;k}X_{h;j}=0 \  \text{ if } j\neq k; \quad X_{h;j} Y_h+Y_hX_{h;j} =0 \  \text{ for }j=1,\dots,d, 
\end{align*}
and $X_{h;j}^2=(D^S_{h;j})^2$, $Y_h^2=1$. These properties suggest that 
\[
\tilde H_{\mathrm{KS};h} =\sum_{j=1}^d X_j +m Y
\]
may be considered as a discrete Dirac operator on $\ell^2(h\ze^d)$. 
In particular, we have
\begin{align*}
(\tilde H_{\mathrm{KS};h})^2u(z) &= \sum_{j=1}^d (D_{h;j}^S)^2u(z)+m^2u(z)\\
&=\sum_{j=1}^d (2h)^{-2}(2u(z)-u(z+2he_j)-u(z-2he_j)) +m^2u(z)\\
&=(-\triangle_{2h}+m^2)u(z)
\end{align*}
for $u\in \ell^2(h\ze^d)$. Whereas $\tilde H_{\mathrm{KS};h}$ is a scaler operator, i.e., 
an operator acting on the one-component function space, it still has the 
fermion doubling problem. In order to solve this problem, we transform the operator 
$\tilde H_{\mathrm{KS};h}$ to an operator $H_{\mathrm{KS};h}$ on $[\ell^2((2h)\ze^d)]^{\oplus 2^d}$. 
By doubling the lattice spacing, we reduce the period of the dual space by half, i.e., 
$((2h)\ze^d)'= (2h)^{-1}\torus^d$, and remove the problematic periodic critical points. 
In order to double the lattice spacing, we increase the number of components to $2^d$, 
in the following way: We define the set of indices $\L$ by 
\[
\L=\{0,1\}^d\subset \ze^d, \quad |\L|=2^d, 
\]
and we write $a=(a_1,\dots, a_d)\in\L$, where $a_j\in\{0,1\}$, $j=1,\dots,d$. 
We consider $2^d\times 2^d$-matrices of the form $L=(L_{a,b})_{a,b\in\L}$. 
We denote
\[
[\ell^2((2h)\ze^d)]^{\oplus\L} 
=\bigset{(u_a(z))_{a\in\L}}{u_a\in\ell^2((2h)\ze^d), a\in\L}.
\]
We define a unitary operator $U_h$ : $\ell^2(h\ze^d) \to [\ell^2((2h)\ze^d)]^{\oplus\L}$ as follows:
\[
(U_h u)_a(z) = 2^{-d/2} u(z+ha), \quad z\in (2h)\ze^d, a\in \L,
\]
for $u\in\ell^2(h\ze^d)$. The adjoint operator is given by 
\[
(U_h^* w)(z+ha) = 2^{d/2} w_a(z), \quad z\in (2h)\ze^d, a\in \L,
\]
where $w=(w_a)_{a\in\L}\in \ell^2((2h)\ze^d)]^{\oplus\L}$. 
Now we define the KS-Hamiltonian by
\[
H_{\mathrm{KS};h}=U_h \tilde H_{\mathrm{KS};h} U_h^*. 
\] 
By direct computations, we learn the $(a,b)$-component of
the matrix operator $U_h X_{h;j} U_h^*$ is given by 
\[
(U_h X_{h;j} U_h^*)_{a,b} u_b(z) = \begin{cases} 
(-1)^{s_{j-1}(a)} D^+_{2h;j} u_b(z)\quad &\text{if }b=a-e_j,  \\
(-1)^{s_{j-1}(a)} D^-_{2h;j} u_b(z)\quad &\text{if }b=a+e_j,  \\
0 \quad &\text{otherwise}, 
\end{cases}
\]
for $u_b\in \ell^2((2h)\ze^d)$, $a,b\in\L$, $j=1,\dots,d$. Thus we have 
\[
(H_{\mathrm{KS};h})_{a,b} u_b(z) =\begin{cases} 
(-1)^{s_{j-1}(a)} D^+_{2h;j} u_b(z)\quad &\text{if }b=a-e_j, \\
(-1)^{s_{j-1}(a)} D^-_{2h;j} u_b(z)\quad &\text{if }b=a+e_j, \\
m (-1)^{s_d(a)} u_b(z)\quad &\text{if } a=b, \\
0 \quad &\text{otherwise}, 
\end{cases}
\]
for $u_b\in \ell^2((2h)\ze^d)$, $a,b\in\L$. 

%%%
\subsection{KS-Hamiltonian in the Fourier space and its eigenvalues}\label{subsec-KS2}

At first we note 
\begin{align*}
(H_{\mathrm{KS};h})^2 &= U_h (\tilde H_{\mathrm{KS};h})^2 U_h^* \\
&=U_h (-\triangle_{2h}+m^2) U_h^* = (-\triangle_{2h}+m^2) \mathbf{1}_{|\L|}
\end{align*}
since $-\triangle_{2h}$ acts on each $(2h)\ze^d+ ha$, $a\in\L$. 

For simplicity, we denote $F_h\mathbf{1}_{|\L|}$ on $[\ell^2(h\ze^d)]^{\oplus \L}$ 
by the same symbol $F_h$. We set
\[
\hat H_{\mathrm{KS};h} = F_{2h} H_{\mathrm{KS};h} F_{2h}^*, 
\]
then it is a matrix with multiplication operators as the entries. Namely, if we denote 
the symbols of the forward/backward difference operators by 
\[
d^\pm_{h;j}(\x)= \pm \frac{e^{\pm2\pi ih\x_j}-1}{ih}, \quad j=1,\dots, d, 
\]
then we have 
\[
(\hat H_{\mathrm{KS};h}(\x))_{a,b} =\begin{cases} 
(-1)^{s_{j-1}(a)} d^+_{2h;j}(\x) \quad &\text{if }b=a-e_j, \\
(-1)^{s_{j-1}(a)} d^-_{2h;j}(\x) \quad &\text{if }b=a+e_j, \\
m (-1)^{s_d(a)} \quad &\text{if } a=b, \\
0 \quad &\text{otherwise}. 
\end{cases}
\]
By the above observation, we also have
\[
(\hat H_{\mathrm{KS};h}(\x))^2 = 
\biggpare{\sum_{j=1}^d (2h^2)^{-1}(1-\cos(4\pi h\x_j)) +m^2}\mathbf{1}_{|\L|}.
\]
In particular, this implies the eigenvalues of $\hat H_{\mathrm{KS};h}(\x)$ are given by 
\[
\pm \hat E_{\mathrm{KS};h}(\x) = \pm \biggpare{\sum_{j=1}^d (2h^2)^{-1}(1-\cos(4\pi h\x_j)) +m^2}^{1/2}. 
\]
We recall that $\hat E_{\mathrm{KS};h}(\x)$ is defined on $(2h)^{-1}\torus^d$, and hence it has a unique minimal 
point $\x=0$. In other words, $H_{\mathrm{KS};h}$ has no fermion doubling problem in 
the Fourier space, though it has many components. 

%%%
\subsection{Continuum limit of the KS-Hamiltonian} \label{subsec-KS3}

Now if we take the limit $h\to 0$, at least formally, $D^\pm_{2h;j}(\x) \to D_j$, and hence 
$H_{\mathrm{KS};h}\to H_{\mathrm{KS};0}$ with some differential operator with constant matrix 
coefficients $H_{\mathrm{KS};0}$ on $[L^2(\re^d)]^{\oplus\L}$. 
For $j=1,\dots,d$ and $a,b\in\L$, we set 
\[
(A_j)_{a,b} =\begin{cases} 
(-1)^{s_{j-1}(a)} \quad &\text{if }b=a+e_j \text{ or }b=a-e_j , \\
0 \quad &\text{otherwise}, 
\end{cases}
\]
and 
\[
B_{a,b} = (-1)^{s_d(a)}\d_{a,b}. 
\]
Then we have 
\[
H_{\mathrm{KS};0} = \sum_{j=1}^d A_j D_j +m B
\quad \text{on } [L^2(\re^d)]^{\oplus\L}.
\]
We can actually show $H_{\mathrm{KS};h}$ converges to $H_{\mathrm{KS};0}$ in the generalized 
norm resolvent sense.

\begin{thm}\label{thm-KSConvergence}
For $z\in\co\setminus\re$, 
\[
\Bignorm{(H_{\mathrm{KS};0}-z)^{-1} - J_{2h} (H_{\mathrm{KS};h}-z)^{-1}J_{2h} ^*}_{\mathcal{B}(L^2(\re^d))}\to 0,
\quad \text{as }h\to 0.
\]
Since $J_{2h} $ is an isometry, this also implies 
\[
\Bignorm{J_{2h} ^*(H_{\mathrm{KS};0}-z)^{-1}J_{2h}  - (H_{\mathrm{KS};h}-z)^{-1}}_{\mathcal{B}(\ell^2(h\ze^d))}\to 0,
\quad \text{as }h\to 0.
\]
\end{thm}

\begin{proof}
We first note
\[
\bigabs{d^\pm_{2h;j}(\x)-2\pi\x_j}\leq \frac{(4\pi h\x_j)^2}{2\cdot2h} \leq 4\pi^2 h|\x|^2,
\]
and 
\[
\bigabs{(\hat{H}_{\mathrm{KS};0}(\x)-z)^{-1}}\leq \max_{\pm}\bigabs{(\pm(|2\pi\x|^2+m^2)-z)^{-1}}
\leq C|\x|^{-1},
\]
uniformly in $\x\in\re^d$, $h>0$, where $\hat H_{\mathrm{KS};0}(\x)=\sum_{j=1}^d 2\pi\x_j A_j +mB$. 
We also have 
\[
\bigabs{(\hat{H}_{\mathrm{KS};h}(\x)-z)^{-1}}=\max_{\pm}\bigabs{(\pm\hat E_h(\x)-z)^{-1}}
\leq C|\x|^{-1}
\]
on the support of $\hat\f(h\x)$. Combining these, we have 
\[
\bigabs{(\hat{H}_{\mathrm{KS};h}(\x)-z)^{-1}-(\hat{H}_{\mathrm{KS};0}(\x)-z)^{-1}}
\leq C h
\]
on the support of $\hat\f(h\x)$, and then it is straightforward to show the claim as in 
the proof of Theorem~\ref{thm-Wilson-convergence}, or \cite{NaTad}. See also \cite{C-G-J-2},
Section~3.1. 
\end{proof}

We note $A_1, \dots,A_d$ and $B$ satisfy the following properties as well as 
$\a_1,\dots,\a_d$ and $\b$, i.e., 
\[
A_j A_k +A_k A_j =0 \ \text{if }j\neq k; \quad 
A_j B +B A_j =0, 
\]
and $A_j^2 =B^2 =\mathbf{1}_{|\L|}$, $j=1,\dots, d$. Thus we may consider 
$H_{\mathrm{KS};0}$ as a Dirac operator, but the number of components are not 
necessarily the same as the standard Dirac operators. Namely, if $d=1$, then 
$2^1=2$ and the number of components is the same as the standard one, 
but if $d=2$, then $2^2=4>2$, and if $d=3$ then $2^3=8>4$, and the number 
of components are twice as that of the standard Dirac operators. 
We expect $H_{\mathrm{KS};0}$ is decomposed to a direct sum of the 
standard Dirac operators, and we confirm it for $d\leq 3$ in Section~\ref{sec-examples}.

%%%%%%%%%%%%%%%%%

\section{Examples}\label{sec-examples}\label{sec-examples}

Here we consider KS-Hamiltonians and their continuum limit for $d=1,2$ and $3$. 

\subsection{1 dimensional case}\label{subsec-ex1D}
For $d=1$, the model is transparent and easy to understand. It is also essentially the same model 
discussed in \cite{C-G-J} Section~3.1 as \textit{the 1D forward-backward difference model}. 

At first, we have 
\[
\tilde H_{\mathrm{KS};h} u(x) = \frac{1}{2ih}(u(x+h)-u(x-h))  +(-1)^{x/h}m u(x), 
\quad x\in h\ze, 
\]
for $u\in\ell^2(h\ze)$, and hence 
\[
H_{\mathrm{KS};h} =\begin{pmatrix} m & D^+_{2h,1}\\ D^-_{2h;1} & -m \end{pmatrix}. 
\]
Its eigenvalues are $\pm \sqrt{(2h^2)^{-1}(1-\cos(4\pi h\x))+m^2}$, and the continuum limit is 
\[
H_{\mathrm{KS};0} =\begin{pmatrix} m & D\\ D & -m \end{pmatrix}
=D \s_1 + m \s_3, 
\quad\text{on }L^2(\re), 
\]
where $D=-i\frac{\pa}{\pa x}$, and thus we recover the standard 1D Dirac operator.

%%%
\subsection{2 dimensional case}\label{subsec-ex2D}
If $d=2$, the one component operator is given by 
\[
\tilde H_{\mathrm{KS};h} u(x,y) = D^S_{h;1}u(x,y)+(-1)^{x/h} D^S_{h;2}u(x,y)
 +(-1)^{(x+y)/h}m u(x,y)
\]
for $(x,y) \in h\ze^2$, where $u\in \ell^2(h\ze^2)$. We set 
\begin{align*}
&u_1(x,y)=u(x,y), \quad u_2(x,y)=u(x+h,y+h), \\
&u_3(x,y)=u(x,y+h), \quad u_4(x,y)=u(x+h,y)
\end{align*}
for $(x,y)\in 2h\ze^2$ and $u\in\ell^2(h\ze^2)$, and then we set
$U_hu=(u_j)_{j=1}^4$.  Applying the formula in Section~\ref{subsec-KS1} we have 
\[
H_{\mathrm{KS};h} =\begin{pmatrix} m & 0 & D^-_{2h,1}& D^-_{2h;2}\\ 
0 & m & -D^+_{2h;2} &D^+_{2h;1} \\ D^+_{2h;1} & -D^-_{2h;2} & -m & 0 \\
D^+_{2h;2} & D^-_{2h;1} & 0 & -m \end{pmatrix}, 
\]
and in the continuum limit, we obtain 
\[
H_{\mathrm{KS};0} =\begin{pmatrix} m & 0 & D_1& D_2\\ 
0 & m & -D_2 &D_1 \\ D_1 & -D_2 & -m & 0 \\
D_2 & D_1 & 0 & -m \end{pmatrix}.
\]
This does not look like the standard 2D Dirac operator, but if we set 
\[
M=\begin{pmatrix}
1 & 0 & 1 & 0 \\ i & 0 & -i & 0 \\ 0 & 1 & 0 & 1 \\ 0 & i & 0 & -i 
\end{pmatrix},
\]
then 
\begin{align*}
M^{-1} H_{\mathrm{KS};0} M &= {\small \begin{pmatrix}
m & D_1+iD_2 & 0 & 0\\D_1-iD_2 & -m & 0&0 \\0 & 0& m & D_1-iD_2 \\0  &0 & D_1+iD_2 & -m 
\end{pmatrix}}\\
& =\begin{pmatrix} \overline{D_1\s_1+D_2\s_2+m\s_3} & 0 \\ 0 & D_1\s_1+D_2\s_2+m\s_3 
\end{pmatrix}. 
\end{align*}
Thus we arrive at a direct sum of two standard 2D Dirac operators, 
one of which is simply the complex conjugate. 

%%%
\subsection{3 dimensional case}\label{subsec-ex3D}
For $d=3$, the computations look somewhat complicated. We omit to write down the definition 
of $\tilde H_{\mathrm{KS};h}$. We set 
\begin{align*}
&u_1(x,y,z)=u(x,y,z), \quad u_2(x,y,z)= u(x+h,y+h,z), \\
&u_3(x,y,z)=u(x,y+h,z+h), \quad u_4(x,y,z)=u(x+h,y,z+h) \\
&u_5(x,y,z)=u(x,y,z+h), \quad u_6(x,y,z)= u(x+h,y+h,z+h), \\
&u_7(x,y,z)=u(x,y+h,z), \quad u_8(x,y,z)=u(x+h,y,z) 
\end{align*}
for $(x,y,z)\in 2h\ze^3$, $u\in \ell^2(h\ze^3)$, and we set $U_hu=(u_j)_{j=1}^8$. 
Then we have 
\[
H_{\mathrm{KS};h}=
{\small \begin{pmatrix}
   m&    0&    0&    0& D^-_3&    0&  D^-_2&  D^-_1\\
   0&    m&    0&    0&   0&  D^-_3&  D^+_1& -D^+_2\\
   0&    0&    m&    0& D^+_2&  D^+_1& -D^+_3&    0\\
   0&    0&    0&    m& D^+_1& -D^-_2&    0& -D^+_3\\
 D^+_3&    0&  D^-_2&  D^-_1&  -m&    0&    0&    0\\
   0&  D^+_3&  D^+_1& -D^+_2&   0&   -m&    0&    0\\
 D^+_2&  D^-_1& -D^-_3&    0&   0&    0&   -m&    0\\
 D^+_1& -D^-_2&    0& -D^-_3&   0&    0&    0&   -m
 \end{pmatrix}, }
\]
and the continuum limit is 
\[
H_{\mathrm{KS};0}=
{\small \begin{pmatrix}
   m&    0&    0&    0& D_3&    0&  D_2&  D_1\\
   0&    m&    0&    0&   0&  D_3&  D_1& -D_2\\
   0&    0&    m&    0& D_2&  D_1& -D_3&    0\\
   0&    0&    0&    m& D_1& -D_2&    0& -D_3\\
 D_3&    0&  D_2&  D_1&  -m&    0&    0&    0\\
   0&  D_3&  D_1& -D_2&   0&   -m&    0&    0\\
 D_2&  D_1& -D_3&    0&   0&    0&   -m&    0\\
 D_1& -D_2&    0& -D_3&   0&    0&    0&   -m
 \end{pmatrix}. }
\]
By setting 
\[
M={\small \begin{pmatrix}
1&  0& 0&  0&  1& 0&  0& 0\\
 i&  0& 0&  0& -i& 0&  0& 0\\
 0&  1& 0&  0&  0& 1&  0& 0\\
 0& -i& 0&  0&  0& i&  0& 0\\
 0&  0& 1&  0&  0& 0&  1& 0\\
 0&  0& i&  0&  0& 0& -i& 0\\
 0&  0& 0&  1&  0& 0&  0& 1\\
 0&  0& 0& -i&  0& 0&  0& i
 \end{pmatrix}, }
\]
we have 
\begin{align*}
&M^{-1} H_{\mathrm{KS};0} M \\
&={\scriptsize \begin{pmatrix}
          m&            0&         D_3& -i D_1 + D_2&            0&           0&            0&           0\\
           0&            m& i D_1 + D_2&         -D_3&            0&           0&            0&           0\\
         D_3& -i D_1 + D_2&          -m&            0&            0&           0&            0&           0\\
 i D_1 + D_2&         -D_3&           0&           -m&            0&           0&            0&           0\\
           0&            0&           0&            0&            m&           0&          D_3& i D_1 + D_2\\
           0&            0&           0&            0&            0&           m& -i D_1 + D_2&        -D_3\\
           0&            0&           0&            0&          D_3& i D_1 + D_2&           -m&           0\\
           0&            0&           0&            0& -i D_1 + D_2&        -D_3&            0&          -m
\end{pmatrix}}\\
&= {\footnotesize \begin{pmatrix}
m\mathbf{1}_2 & D_1 \s_2+ D_2\s_1 +D_3\s_3 & 0 & 0 \\
D_1 \s_2+ D_2\s_1 +D_3\s_3 & - m\mathbf{1}_2 & 0 & 0\\
0 & 0 & m\mathbf{1}_2 & D_1 \overline{\s_2}+ D_2\s_1 +D_3\s_3 \\
0 & 0 & D_1 \overline{\s_2}+ D_2\s_1 +D_3\s_3 & -m\mathbf{1}_2
\end{pmatrix}}\\
&={ \begin{pmatrix}
D_1\a_2+D_2\a_1+D_3\a_3+m\b & 0 \\
0 & \overline{D_1\a_2+D_2\a_1+D_3\a_3+m\b} . 
\end{pmatrix}}
\end{align*}
This gives a direct sum of a representation of 3D Dirac operator and its complex conjugate
(another representation). Of course, the final form depends on the choice of the 
diagonalization matrix $M$. These matrix computations were aided by a symbolic computation tool
\textsf{SymPy} on \textsf{Python}. 

%%%%%

\end{document}